\documentclass[aps,apl,preprint,a4paper,amsmath,amssymb,superscriptaddress,floatfix,nofootinbib]{revtex4-1}
\usepackage{amsmath}         % features for displayed equations
\usepackage{amssymb}         % additional maths symbols
\usepackage{amsfonts}        % additional maths fonts
\usepackage{amsthm}          % theorem environments setup
\usepackage[final]{hyperref} % hypertext links
\usepackage{graphicx}        % extended graphics package
\usepackage{subfig}          % subfigures
\usepackage[UKenglish]{babel}% correct hyphenation rules for UK English
\usepackage[T1]{fontenc}     % better treatment of accented words
\usepackage{todonotes}
\graphicspath{{figures/}}    % store figures in /figures/ subdirectory
\newcommand{\Tr}{\text{Tr}}							%Trace%
\newcommand{\id}{\text{d}}\,						%Infinitesimal, straight 'd'%
\newcommand{\set}[1]{ \left\lbrace #1 \right\rbrace }	%Set%
\newcommand{\bra}[1]{\langle #1 \vert}
\newcommand{\ket}[1]{\vert #1 \rangle}
\newcommand{\mat}[4]{\left(\begin{array}{cc} #1 & #2 \\ #3 & #4 \end{array}\right)}%
\theoremstyle{definition}
\newtheorem{lem}{Lemma}
\newtheorem{thm}{Theorem}
\newtheorem{col}{Corollary}
\newtheorem{definition}{Definition}
\begin{document}
\title{Sanov and Central Limit Theorems for output statistics of quantum Markov chains}
\author{Merlijn \surname{van Horssen}}
\email[E-mail: ]{merlijn.vanhorssen@nottingham.ac.uk}
\affiliation{School of Physics and Astronomy, University of Nottingham, Nottingham, NG7 2RD, UK}
\author{M\u{a}d\u{a}lin \surname{Gu\c{t}\u{a}}}
\email[E-mail: ]{madalin.guta@nottingham.ac.uk}
\affiliation{School of Mathematical Sciences,\\ University of Nottingham, Nottingham, NG7 2RD, UK}
\date{\today}
\begin{abstract}
In this paper we consider the statistics of repeated measurements on the output of a quantum Markov chain. We establish a large deviations result analogous to Sanov's theorem for the empirical measure associated to finite sequences of consecutive outcomes of a classical stochastic process. Our result relies on the construction of an extended quantum transition operator (which keeps track of previous outcomes) in terms of which we compute moment generating functions, and whose spectral radius is related to the large deviations rate function. As a corollary to this we obtain a central limit theorem for the empirical measure. Such higher level statistics may be used to uncover critical behaviour such as dynamical phase transitions, which are not captured by lower level statistics such as the sample mean. As a step in this direction we give an example of a finite system whose level-one rate function is independent of a model parameter while the level-two rate is not.
\end{abstract}
\maketitle

\section{Introduction}

Quantum Markov processes are effective mathematical models for describing a wide class of open quantum dynamics where the environment interacts weakly with the system \cite{Gardiner2004,Breuer2002,Kummerer2002}. In the discrete time version, a quantum Markov chain consists of a system which interacts successively with identically prepared ancillas (input), via a fixed unitary transformation, cf. Figure \ref{fig.markov}. After the interaction, the ancillas (output) are in a finitely correlated state \cite{Fannes1992} which carries information about the dynamics. Such information can be extracted by performing successive measurements on the outgoing ancillas. Similar models are used in continuous time in the input-output formalism of quantum optics \cite{Gardiner2004}, where output measurements (e.g. photon counting or homodyne) are used to monitor and control the system \cite{Wiseman&Milburn}. Interest in such systems has grown with recent experimental progress in quantum optics and open quantum many-body systems \cite{Konya2012, Diehl2010a, Ates2012, Olmos2012}, in particular from the point of view of dynamical phase transitions \cite{Netocny2004a,Andrieux2010,Garrahan2011,Lee2012,Foss-Feig12,Kessler12,VanHorssen2012,Lesanovsky2013,Malossi2013,Carr2013} and system identification \cite{Guta2011,Guta2014}.

\begin{figure}[h]
\includegraphics[width=.75\textwidth]{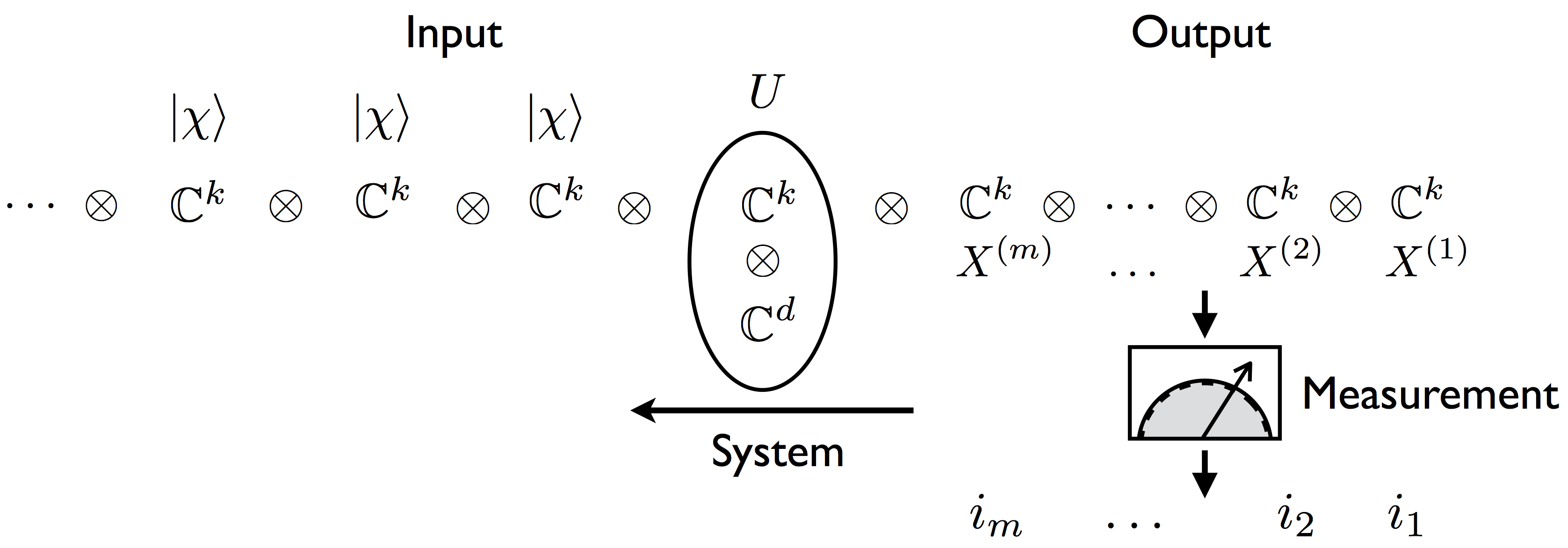}
\caption{A quantum Markov chain consists of a system $\mathbb{C}^d $ interacting successively with a sequence of identically prepared input ancillas with space $\mathbb{C}^k$ via the unitary operator $U$. Performing a measurement of a single-site operator $X$ on each of the output sites results in a sequence of measurement outcomes $i_{1}, i_{2}, \ldots, i_{m}$. }\label{fig.markov}
\end{figure}
Dynamical phase transitions in open quantum systems are visible through spectral properties of the generator of the dynamics \cite{Garrahan2011}, and through the thermodynamics of jump trajectories \cite{Benson1994,Pickles1996,Garrahan2009,Genway2012}; the large deviations approach to quantum phase transitions exploits both these features to uncover dynamical phases \cite{Lesanovsky2013,VanHorssen2014}. In this article we establish a large deviation principle (LDP) for \emph{counts of sequences of successive outcomes} in the output of a quantum Markov chain, which we refer to as the Sanov theorem for such systems. Sanov theorems have been considered in quantum systems before, both generally \cite{Hiai2007} and in the context of quantum hypothesis testing \cite{Bjelakovic2004, Hiai2008a, Bjelakovic2008, Jaksic2012}; our result extends previous work on LDPs for counting statistics, by considering counts of sequences of outcomes on an arbitrary number of subsequent sites.

The Sanov theorem for classical Markov chains \cite{DenHollander2000, Dembo2010} establishes an LDP for the empirical measure (the proportion of states visited by the chain) and the pair empirical measure (the frequency with which jumps between pairs of states occur) (see Fig. \ref{fig:sanovevents}); a result used for example in \cite{Ellis1995} to characterise phase transitions in the Curie-Weiss-Potts model. In non-equilibrium statistical mechanics, a dynamical phase transition may occur when the G\"{a}rtner-Ellis theorem for some statistic fails to hold \cite{Ellis1995, Touchette2009}, corresponding to a non-analyticity of the associated LD rate function.

To establish the LDP for the higher-level statistics in the Sanov theorem, we express the associated sequence of moment generating functions in terms of an extended transition operator, constructed from the original quantum Markov chain transition operator (analogous to the method in \cite{Hiai2007}). The LDP is then obtained via the G\"{a}rtner-Ellis theorem; the corresponding LD rate function is obtained in terms of the spectral radius of a perturbation of this new transition operator, where primitivity of the original transition operator ensures that this rate function is analytic.

In general, the existence of an LDP for a stochastic process does not imply that the process also satisfies a central limit theorem. However, in certain cases \cite{Letters1993} the proof of the LDP may be extended to produce a central limit theorem; as a corollary to our main result, we thus state a central limit theorem for the empirical measures.

The paper is organised as follows: in Sec. \ref{sec:bg} we introduce the background to our result. We briefly review the theory of large deviations in Sec. \ref{sec:bg-ld}, we introduce quantum Markov chains in Sec. \ref{sec:bg-qmc} and consider measurements on the output of a quantum Markov chain in Sec. \ref{sec:bg-output}. In Sec. \ref{sec:result} we state our main result, Thm. \ref{thm:sanov}, and Corollary \ref{cor:clt} (preceded by definitions and results directly related to our main result). To illustrate our result we discuss two examples in Sec. \ref{sec.examples}.

%%%%%%%%%%%%%%%%%%%%%%
\section{Background}\label{sec:bg}
%%%%%%%%%%%%%%%%%%%%%%
In this section we present a brief review of the basic concepts of large deviations needed in this paper, and introduce the set-up of quantum Markov chains. For good introductions to the theory and applications of large deviations we refer to the monographs \cite{Varadhan1984, Ellis1995, DenHollander2000, Dembo2010}.

%%%%%%%%%%%%%%%%%%%%%%%
\subsection{Large deviations}\label{sec:bg-ld}
%%%%%%%%%%%%%%%%%%%%%%%

Let $Y_1, Y_2,\dots $ be a sequence of independent and identically distributed (i.i.d.) $\mathbb{R}^d$-valued random variables defined on a probability space $(\Omega, \Sigma, \mathbb{P})$.
%, and let $\mathbb{P}_Y$ denote their common probability distribution on $\mathbb{R}^d$. 
The law of large numbers (LLN) states that if $\mathbb{E}(Y_i)=y$ exists, then the average 
$$
X_n := \frac{1}{n} \sum_{i=1}^n Y_i
$$ 
converges to $y$ almost surely as $n\to \infty$. The Central Limit Theorem (CLT) characterises the speed of convergence, and shows that the fluctuations around the mean decrease as 
$n^{-1/2}$, and are asymptotically normally distributed: 
$$
\sqrt{n}(X_n-y) \overset{\mathcal{L}}{\longrightarrow} N(0,v).   
$$
Here $\mathcal{L}$ denotes the convergence in law (distribution) and $N(0,v)$ is the normal distribution with mean zero and variance $v:= \mathbb{E}(Y_i^2) -\mathbb{E}(Y_i)^2$. As the distribution of $X_n$ concentrates around $y$, one would like to know how the probability of staying away from $y$ decreases with $n$. Such large deviations have exponentially small probabilities i.e.
\begin{equation}\label{eq.ld.iid}
\mathbb{P}(X_n  \geq y + a ) \sim \exp(-n I(a) ) 
\end{equation}
in a sense which will be made precise below. 

\begin{figure}
    \subfloat[]
    {\includegraphics[width=.25\textwidth]{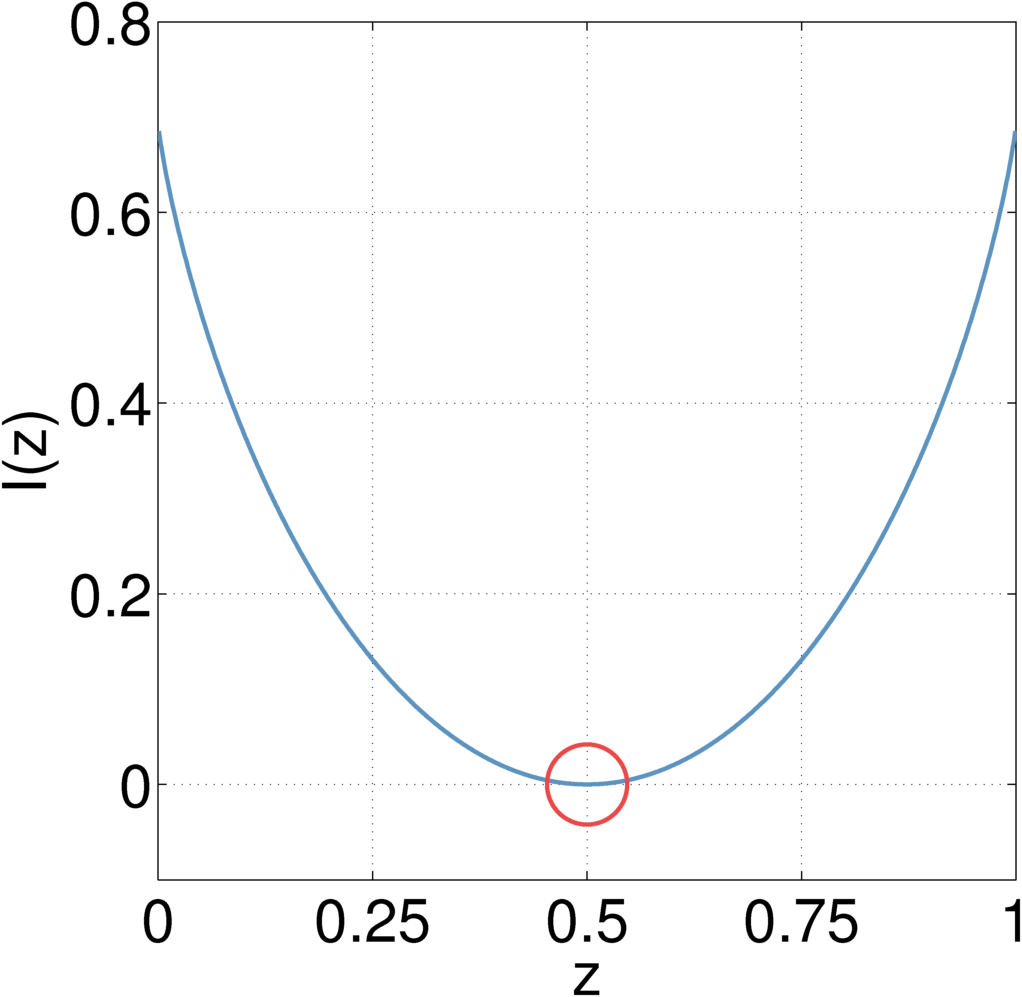}}\quad
    \subfloat[]   
    {\includegraphics[width=.26\textwidth]{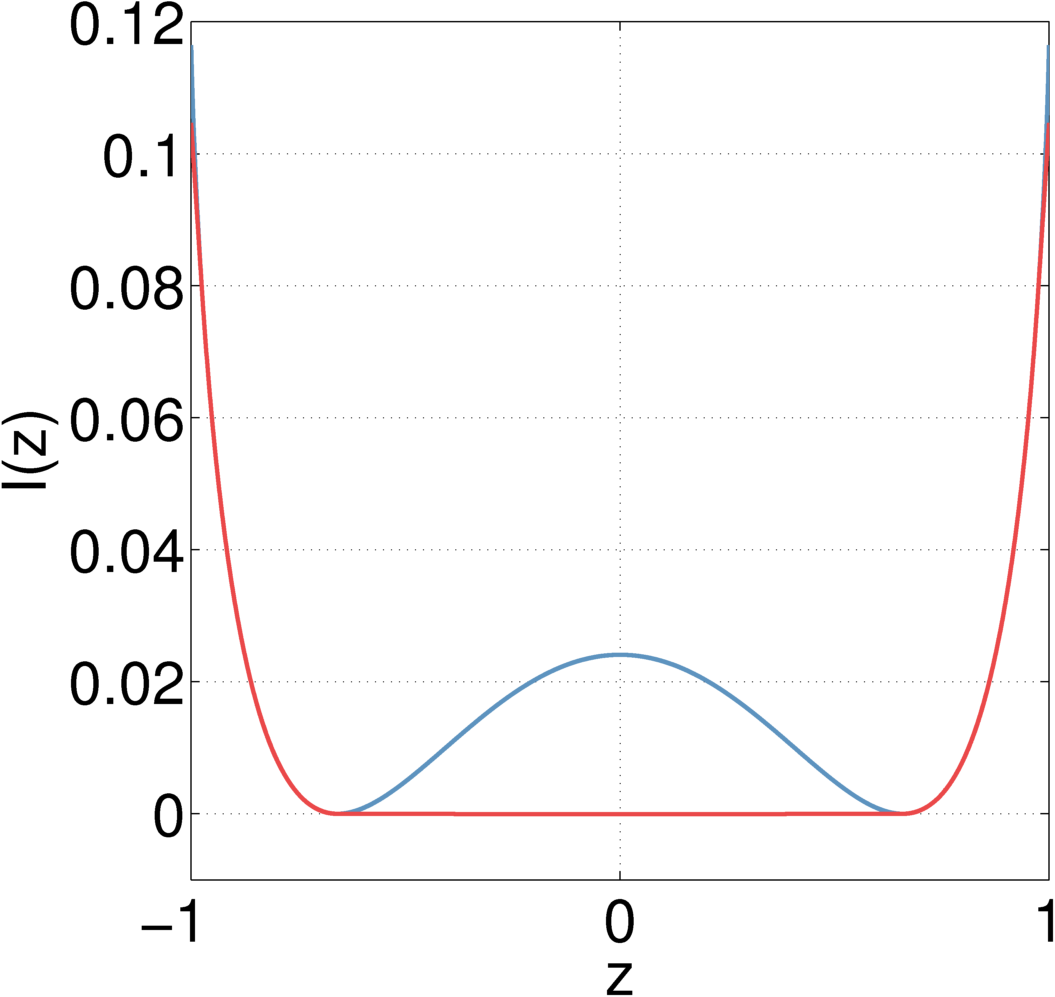}}\quad
        \subfloat[]
    {\includegraphics[width=.395\textwidth]{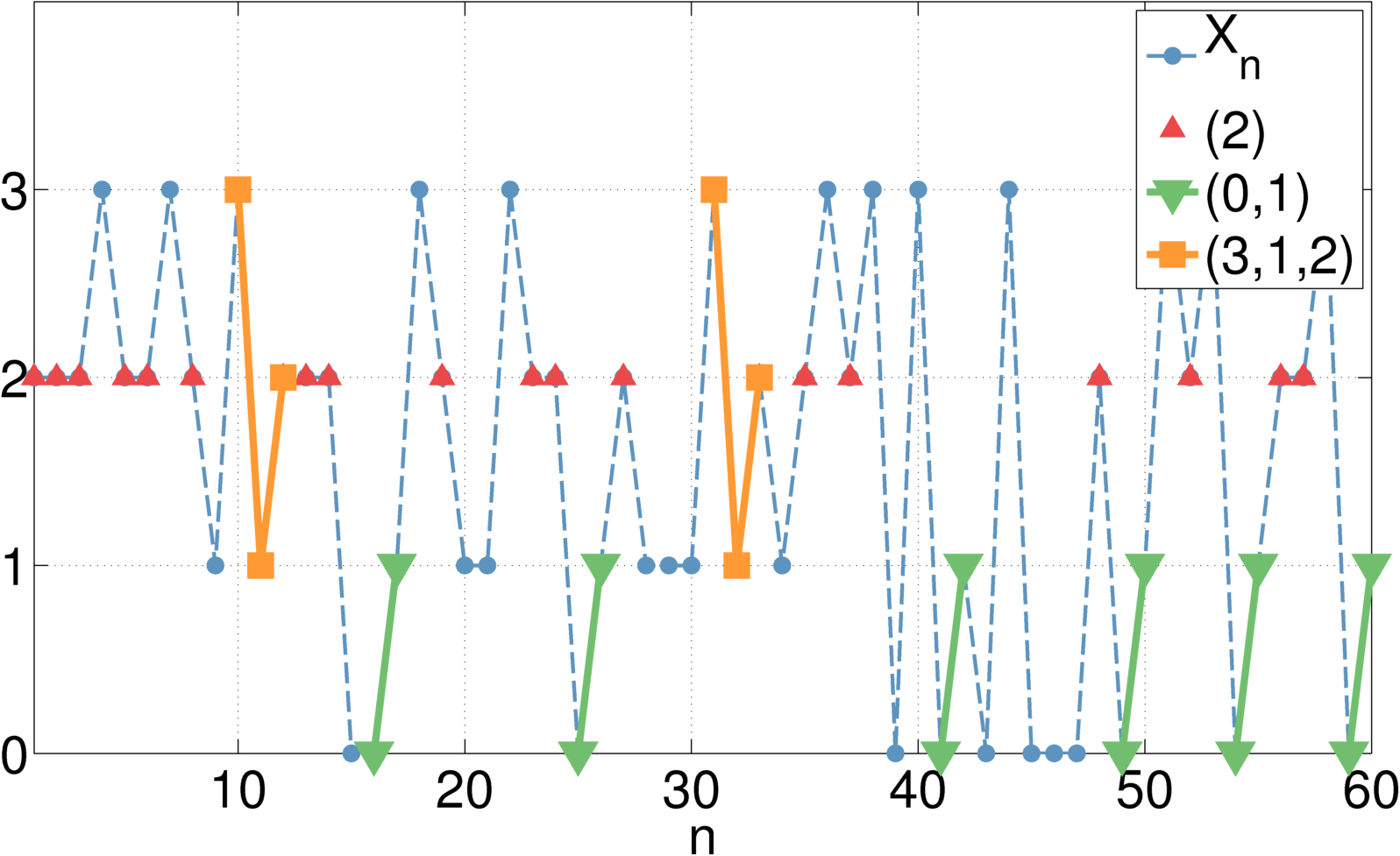}\label{fig:sanovevents}}     
    \caption{Typical large deviations rate functions. \textbf{(a):} rate function $I(z)$ on the interval $[0,1]$ associated to the sample mean of a fair coin toss, with minimum of $I$ indicated at $z = 1/2$; \textbf{(b):} example of a non-convex rate function (blue) and its convex envelope (red) obtained from the G\"{a}rtner-Ellis theorem. \textbf{(c):} Sample trajectory $X_{n}$ with events associated to different levels, e.g. visits to state $2$, jumps from $0$ to $1$, etc., which are used to compute the empirical measures.}
    \label{fig:ldrate}
\end{figure}

In another example, suppose $Y_1, Y_2, \dots$ are i.i.d. random variables with values in $\{1,\dots, d\}$ and common probability distribution $p_i = \mathbb{P}(Y= i)$. The \emph{empirical measure}  $\hat{\mathbb{P}}_n$ is defined as the random probability distribution given by the frequencies with which the different values in $\{1,\dots ,d\}$ occur (where the indicator function is defined by $1_{i}(Y) = 1$ if $Y=i$, and $0$ otherwise)
\begin{equation}\label{eq.ld.sanov}
\hat{\mathbb{P}}_n (i)= \frac{1}{n} \sum_{j=1}^n 1_{i}(Y_j).
\end{equation}
Again, by the LLN we have $\hat{\mathbb{P}}_n (i) \to p_i$, and the fluctuations around the mean can be described by the CLT. The large deviations are characterised by the following result known as the Sanov theorem. If $R$ is a measurable subset of the simplex of probability distributions, and $p \notin R$, then
$$
\mathbb{P} (\hat{\mathbb{P}}_n \in R) \sim \exp \left( -n  \inf_{q\in R} I (q|p) \right)
$$
where the rate function $I(q|p)$ is the relative entropy 
$$
I(q|p)= \sum_{i=1}^d p_i \log (p_i/q_i).
$$
We summarise \eqref{eq.ld.sanov} by saying that "a large deviation event will happen in the least unlikely of all the unlikely ways" \cite{DenHollander2000}.

We can now describe a more general set-up of large deviations theory \cite{Varadhan1984, Ellis1995, DenHollander2000, Dembo2010}, and formulate one of the key mathematical tools used later in our paper. A sequence $\set{\mu_{n}:n \in \mathbb{N}}$ probability distributions on $\mathbb{R}^d$ is said to satisfy a \emph{large deviation principle} (LDP) if there exists a lower semicontinuous function $I : \mathbb{R}^{d} \rightarrow [0,\infty]$, called a \emph{rate function}, such that for all measurable subsets $B \subset \mathbb{R}^{d}$
\begin{eqnarray}
		-\displaystyle \inf_{x \in B^{0}} I(x) &\leq & \displaystyle \liminf_{n \rightarrow \infty} \tfrac{1}{n} \log \mu_{n}(B)\nonumber\\ 
		&\leq & \displaystyle \limsup_{n \rightarrow \infty} \tfrac{1}{n} \log \mu_{n}(B) \leq -\displaystyle \inf_{x \in \bar{B}} I(x).\nonumber
\end{eqnarray}
Here $B^{0}$ and $\bar{B}$ denote the interior and closure of $B$, respectively; if these sets coincide, the LDP may be expressed in the intuitive form used above $\mu_{n}(B) \sim \exp \left( -n \inf_{x \in B} I(x) \right)$.

In our results we will employ a theorem due to G\"{a}rtner~\cite{Gartner1977} and Ellis~\cite{Ellis1984}, which provides a necessary condition for obtaining an LDP; we state it here with stronger assumptions, equivalent to the result in \cite{Gartner1977}.
\begin{thm}\label{th.ge}
Let  $\set{\Gamma_{n} : n \in \mathbb{N}}$ be the moment generating function associated to $\mu_n$
	\[	\Gamma_{n}(t) = \int_{\mathbb{R}^{d}} e^{n \langle t, x \rangle} \id \mu_{n}(x),\quad t \in \mathbb{R}^{d},	\]
where $\langle \cdot,  \cdot \rangle$ denotes the standard Euclidean inner product. Suppose the limit
\begin{equation}\label{eq:gelimit}
	F(t) = \lim_{n \rightarrow \infty} \tfrac{1}{n} \log \Gamma_{n}(t)
\end{equation}
is finite for all $t \in \mathbb{R}^{d}$ and $F : \mathbb{R}^{d} \rightarrow \mathbb{R}$ is a differentiable function. Then $\set{\mu_{n} : n \in \mathbb{N}}$ satisfies an LDP with rate function $I$ given by the Legendre-Fenchel transform of $F$,
\begin{equation}\label{eq:gerate}
	I(x) = \sup_{t \in \mathbb{R}^{d}} \set{ \langle t, x \rangle - F(t)}.
\end{equation}
\end{thm}

%\begin{figure} [t]  
%    \subfloat
%    {\includegraphics[width=.5\textwidth]{sanovillustration}}
%    % sanovillustration(60,[2],[0 1],[3 1 2])
%\caption{Sample trajectory $X_{n}$ with events associated to different levels, e.g. visits to state $2$, jumps from $0$ to $1$, etc., which are used to compute the empirical measures.}
%
%\end{figure}

The Gartner-Ellis theorem implies that the rate function $I(\cdot)$ in \eqref{eq.ld.iid} is given by the Legendre transform of 
$\Gamma(t) = \mathbb{E}(e^{tY})$, and can also be employed for proving the Sanov theorem, where $\mu_n$ is the distribution of the empirical measure $\hat{\mathbb{P}}_n$ seen as a vector in $\mathbb{R}^d$. 

The Sanov Theorem can be extended \cite{Sanov1957,Miller1961,Donsker1976} to empirical measures associated to an \emph{irreducible}\footnote{A Markov chain with transition matrix $\Pi$ is called irreducible if for all $1 \leq i,j \leq d$ there is an integer $n$ such that $\Pi^{n}(i,j) > 0$} Markov chain $\set{X_{n} : n \in \mathbb{N}}$ over a discrete state space $\{1,\dots ,d\}$ with transition matrix $\Pi$. For instance, the empirical measure $\hat{\mathbb{P}}^{(1)}(i) := \sum_{j=1}^n 1_{i}(X_j)$ keeps track of the empirical frequencies associated to each state which by ergodicity converge the stationary distribution of the chain. This empirical measure satisfies an LDP on $\mathbb{R}^{d}$, where the rate function is given by the Legendre transform of the function $\lambda \mapsto \log r(\Pi_{\lambda})$. Here $r(\cdot)$ denotes the \emph{spectral radius}; the matrix $\Pi_{\lambda}$ is a certain analytic perturbation of $\Pi$. Similarly, the pair-empirical measure
$$
\hat{\mathbb{P}}^{(2)}(k,l) := \sum_{j=1}^{n-1} 1_{k,l}(X_j, X_{j+1}),
$$
which encodes additional information about how the chain jumps from one state to another, also satisfies an LDP. These empirical measures of different orders are illustrated in  Fig.\ref{fig:sanovevents} on a sample trajectory of a Markov chain with four states. We note that this approach used to establish LDPs for the empirical measures associated to Markov chains bears similarities to the proof of the main result of this paper which deals with the output process of a quantum Markov chains.

\subsection{Quantum Markov chains}\label{sec:bg-qmc}

A quantum Markov chain \cite{Kummerer2002} consists of of a system, or `memory' with Hilbert space $\mathbb{C}^d$ which interacts successively (moving from right to left) with a chain of identically prepared ancillas, or `noise units' $\mathbb{C}^k$, via a unitary 
$U:\mathbb{C}^d\otimes\mathbb{C}^k\to \mathbb{C}^d\otimes\mathbb{C}^k$, cf. Figure \ref{fig.markov}. Physically, this may be seen a discrete-time model for the evolution of an open quantum system coupled to the environment in the Markov approximation, as shown in \cite{Buchleitner2002,Attal2006}. 

We assume that the initial state of the noise units (or input) is a fixed vector $|\chi\rangle\in \mathbb{C}^d$, such that the output state is  determined by the isometry
\begin{eqnarray}
V&:&\mathbb{C}^d\to \mathbb{C}^d\otimes \mathbb{C}^k\nonumber\\
V&:& \ket{\psi}  \mapsto U( \ket{\psi} \otimes |\chi\rangle ) = \sum_{i=1}^{k} V_{i} \ket{\psi} \otimes \ket{i}\label{eq:unitarykraus},
\end{eqnarray}
where $\{|1\rangle, \dots ,|k\rangle \}$ are the vectors of an orthonormal basis in $\mathbb{C}^k$.
%\begin{figure}[t]
%\includegraphics[width=.75\textwidth]{markov}
%\caption{A quantum Markov chain consists of a system $\mathbb{C}^d $ interacting successively with a sequence of identically prepared input ancillas with space $\mathbb{C}^k$ via the unitary operator $U$. Performing a measurement of a single-site operator $X$ on each of the output sites results in a sequence of measurement outcomes $i_{1}, i_{2}, \ldots, i_{m}$. }\label{fig.markov}
%\end{figure}
The operators $V_i = \langle i|U|\chi\rangle$ will be called the \emph{Kraus operators} associated to the quantum Markov chain, and satisfy the normalisation condition $\sum_{i=1}^{k} V_{i}^{\ast} V_{i} = \mathbf{1}$. If the system starts in the state $|\psi\rangle$, we can apply  \eqref{eq:unitarykraus} successively to express the \emph{joint} state of $n$ output units and the system as a matrix product state \cite{Schon2005,Schoen2006}
\begin{equation} \label{eq:sanovstate}
	\vert \psi^{(n)} \rangle = \sum_{i_{1}, \ldots, i_{n} = 1}^{k} V_{i_{n}} \cdots V_{i_{1}} \vert \psi \rangle \vert i_{n}, \ldots, i_{1} \rangle
\end{equation}
reflecting the inherent Markovian character of the dynamics. By tracing over the noise units we find that the reduced system dynamics is given by the semigroup 
$$
T^n_*: M_d \to M_d , \qquad n\in \mathbb{N},
$$ where 
\begin{equation*}
    T_{\ast}:\rho \mapsto \sum_{i=1}^{k} V_{i} \rho V_{i}^{\ast} =  \Tr_{\mathbb{C}^{k}}\left(U\, \rho \otimes |\chi\rangle\langle \chi |\, U^{\ast} \right)
\end{equation*}
is a trace preserving completely positive map describing the system's transition operator. For later purposes we note that 
the corresponding map  in the Heisenberg representation is given by
\begin{eqnarray*}
T&:& M_d \to M_d\\
T&:& A \mapsto \sum_{i=1}^{k} V_{i}^{\ast}A V_{i}.
\end{eqnarray*}
Now, suppose that after the interaction we perform a projective measurement on each of the output noise units, with respect to the basis $\{|1\rangle ,\dots, |k\rangle\}$. If $X^{(i)}$ denotes the outcome of the measurement of the $i$th unit, then the joint probability distribution of the measurement process is 
$$
p(i_1,\dots, i_n) = \mathbb{P}( X^{(1)} = i_1, \dots, X^{(n)}= i_n ) = \| V_{i_{n}} \cdots V_{i_{1}} \ket{\psi}\|^2. 
$$ 

This process is not necessarily stationary but becomes so in the large $n$ limit if $T$ satisfies a certain ergodic property  discussed in Sec. \ref{sec:result} below.

\subsection{Empirical measures associated to the measurement process}\label{sec:bg-output}

%For the rest of the paper we will assume that $T$ is primitive.
%, and that we are in the stationary regime, which can be achieved by starting the system in the stationary state $\rho_{\text{ss}}$. 
Our main goal is to establish a large deviation principle for the empirical measure associated to chains of subsequent outcomes occurring in the measurement trajectory. For each $m\in \mathbb{N}$ we define the empirical measure 
$\hat{\mathbb{P}}^{(m)}_n $ over $\{1,\dots, k\}^m$ by
\begin{equation}\label{eq.empirical.measure}
\hat{\mathbb{P}}^{(m)}_{n}(i_1,\dots, i_m) := \frac{1}{n-m+1}\sum_{j=1}^{n-m+1} 1_{(i_1,\dots, i_m)}(X^{(j)}, \dots , X^{(j+m-1)}).
\end{equation}
In the case $m=1$ this keeps track of the frequencies of the different outcomes in a measurement trajectory, and the associated LDP has been established in \cite{Hiai2007}. A similar result holds in continuous time, for total counts statistics when a counting measurement is performed on the output, and this plays a key role in the theory of \emph{dynamical phase transitions} \cite{Garrahan2009,Garrahan2011}. For $m=2$ the empirical measure captures the statistics of pairs of subsequent outcomes. As we will show in Sec. \ref{sec.examples}, the higher level $(m > 1)$ LD rates capture more information about the measurement process than the total counts, and therefore could be the basis of a more in-depth understanding of the theory of dynamical phase transitions.

The main tool in establishing the LDP will be the G\"{a}rtner-Ellis theorem \ref{th.ge}. For this we need to compute the moment generating function 
\begin{equation}\label{eq.mgf}
 \Gamma_{n}^{(m)}(t):= \mathbb{E} \left( \exp( n \langle \hat{\mathbb{P}}^{(m)}_{n} , t\rangle  )\right)
\end{equation}
where $t\in \mathbb{R}^{k^m}$ and 
$$
\langle \hat{\mathbb{P}}^{(m)}_{n} , t\rangle = \sum_{i_1,\dots ,i_m} t_{i_1,\dots, i_m }\hat{\mathbb{P}}^{(m)}_{n}(i_1,\dots, i_m).
$$
We will show that the MGF can be expressed in terms of a certain extended transition operator which acts on the system but also takes into account the results of $m-1$ measurement outcomes. For this we consider the space  
$$
\mathcal{D}^{m,k}_{d}:=M_{d} \otimes \left( \mathbb{C}^{k} \right)^{\otimes (m-1)},
$$ 
seen as a block diagonal algebra with $k(m-1)$ blocks isomorphic to $M_{d}$.  We will represent an element $Y\in \mathcal{D}^{m,k}_{d}$ by the diagonal elements 
$\left[ Y \right]_{i_{1}, \ldots, i_{m-1}}\in M_d$, where $i_{1},\dots , i_{m-1}\in \{1,\dots ,k\}$.

\begin{lem}\label{lemma.mgf.semigroup}
Let 
$
T_{t,m} :  \mathcal{D}^{m,k}_{d}  \rightarrow \mathcal{D}^{m,k}_{d}
$
be the completely positive map defined as 
\begin{equation}
\left[ T_{t,m}(Y) \right]_{i_{1},\ldots,i_{m-1}} = \sum_{i_{m}=1}^{k} e^{t_{i_{1},\ldots,i_{m}}}
V_{i_{1}}^{\ast} \left[ Y \right]_{i_{2},\ldots,i_{m}} V_{i_{1}}
\label{eq:sanovmaps}
\end{equation}
Then the moment generating function \eqref{eq.mgf}  can be expressed as
\[
    \Gamma_{n}^{(m)}(t) = \langle \psi \vert \sum_{i_{1}, \ldots, i_{m-1} =1}^{k} \left[ T_{t,m}^{n-m+1}\left(M^{(m)}\right)\right]_{i_{1}, \ldots, i_{m-1}} \vert \psi \rangle
\]
where
\begin{equation*}
    \left[ M^{(m)} \right]_{i_{1}, \ldots, i_{m-1}} = V_{i_{1}}^{\ast} \cdots V_{i_{m-1}}^{\ast} V_{i_{m-1}} \cdots V_{i_{1}}.
\end{equation*}
In particular, $M^{(m)}$ is an eigenvector of $T_{0,m} $ with eigenvalue $1$.

\end{lem}

\emph{Proof.} We denote by $X^{(l,m)}$ the $\mathbb{R}^{k^m}$-valued random variable which represents the outcomes on $m$ subsequent sites: for a sequence of outcomes $\mathbf{i} = (i_{1}, \ldots, i_{m})$ we have
\[ 
    \left[ X^{(l,m)} = \mathbf{i}\right] \Leftrightarrow \left[ X^{(l)} = i_{1}, \ldots, X^{(l+m-1)} = i_{m}\right]
\]
where $X^{(j)}$ is the random variable associated to measurement outcomes of the observable $X$ on site $j$. The moment generating function is then
\begin{align*}
    \Gamma_{n}^{(m)}(t) &= \mathbb{E} \left[ \exp  \left\langle t, \sum_{j=1}^{n-m+1} X^{(j,m)}   \right\rangle \right].
\end{align*}
Writing $P^{(l,m)}_{\mathbf{i}}:= P^{(l)}_{i_{1}} \otimes \cdots \otimes P^{(l+m-1)}_{i_{m}}$ for the one dimensional 
projection on $m$ subsequent sites associated to a sequence of outcomes $X^{(l,m)} = \mathbf{i}$, we have
\begin{equation*}
	\Gamma^{(m)}_{n}(t) =  \left\langle  \exp \left( \sum_{l=1}^{n-m+1} \sum_{\mathbf{i}} t_{\mathbf{i}} P^{(l,m)}_{\mathbf{i}}\right) \right\rangle,
\end{equation*}
where the expectation is taken with respect to the system-output state $|\psi^{(n)}\rangle$ defined in \eqref{eq:sanovstate}.
This gives 
\begin{equation} 
    \Gamma^{(m)}_{n}(t) =  \sum_{i_{1}, \ldots, i_{n}=1}^{k} \langle \psi \vert V_{i_{1}}^{\ast} \cdots V_{i_{n}}^{\ast} V_{i_{n}} \cdots V_{i_{1}}  \vert \psi \rangle
    \cdot \exp \left( \sum_{l=1}^{n-m+1} t_{i_{l}, \ldots, i_{l+m-1}} \right).\label{eq:sanovlevelmmgf}
\end{equation}
Finally, a short computation shows that $\Gamma_{n}^{(m)}$ can be expressed as
\[
    \Gamma_{n}^{(m)}(t) = \langle \psi \vert \sum_{i_{1}, \ldots, i_{m-1} =1}^{k} \left[ T_{t,m}^{n-m+1}\left(M^{(m)}\right)\right]_{i_{1}, \ldots, i_{m-1}} \vert \psi \rangle
\]
where
\begin{equation*}
    \left[ M^{(m)} \right]_{i_{1}, \ldots, i_{m-1}} = V_{i_{1}}^{\ast} \cdots V_{i_{m-1}}^{\ast} V_{i_{m-1}} \cdots V_{i_{1}}.
\end{equation*}

The eigevalue property
$$
T_{0,m} (M^{(m)})=  M^{(m)}
$$
follows directly from the definition of $T_{t,m} $ and the normalisation $\sum_i V_i^* V_i = \mathbf{1}$.

\qed

%%%%%%%%%%%%%%%%%%%%%%%%%%%%%%%%%%%%%%%%%%
\section{Main result}\label{sec:result}
%%%%%%%%%%%%%%%%%%%%%%%%%%%%%%%%%%%%%%%%%%%

In this section we recall notions of irreducibility and primitivity  (Def. \ref{def:irred}), and state existing results (Thm. \ref{thm:qpf} and Lem. \ref{thm:sanovlem}) concerning irreducible maps. We then state and prove our main result, the Sanov theorem for the empirical measure \eqref{eq.empirical.measure} in Thm. \ref{thm:sanov}. We note that the following results are all in the context of positive maps on finite-dimensional $C^{\ast}$-algebras. In our theorem, we are considering positive linear maps on algebras with a block form $\bigoplus_{j} M_{d_{j}}$, where each $M_{d_{j}}$ is an algebra of $d_{j} \times d_{j}$ matrices with complex entries; algebras of this form are a particular class of $C^{\ast}$-algebras \cite{Dixmier1981}.

\begin{definition}[\cite{Evans1978}]\label{def:irred}
Let $\mathcal{A}$ be a finite-dimensional $C^{\ast}$-algebra and let $R$ be a positive linear map on $\mathcal{A}$. Then is $R$ called 
\begin{itemize}
\item[(i)]
\emph{irreducible} if there exists $n\in \mathbb N$ such that $({\rm Id}+R)^n$ is strictly positive, 
i.e.  
$
({\rm Id}+R)^n(A) >0 
%\langle \psi|T^n(|\varphi\rangle\langle\varphi|)\psi\rangle>0
$
for all positive operators $A$, and 

\item[(ii)] \emph{primitive} if there exists an $n\in \mathbb N$, such that
$
R^n 
%\langle \psi|T^n(|\varphi\rangle\langle\varphi|)\psi\rangle>0
$
is strictly positive. % unit vectors $\psi,\varphi\in \hi$. 
\end{itemize}
%We call an isometry $V:\hi\to\hi\otimes\hik$ irreducible (primitive), if the associated channel defined in \eqref{eq.T&V} is irreducible (primitive).
\end{definition}
Primitivity is a stronger requirement than irreducibility. 
The following theorem (see \cite{Evans1978,Sanz2010,Fagnola2010}) collects the essential properties of irreducible (primitive) quantum transition operators needed in this paper.

\begin{thm}[{\bf quantum Perron-Frobenius} \cite{Evans1978,Sanz2010,Fagnola2010}]\label{thm:qpf}
Let $\mathcal{A}$ be a finite-dimensional $C^{\ast}$-algebra and let $R$ be a positive linear map on $\mathcal{A}$. Denote by $r(R) := \max_i |\lambda_i| $ the \emph{spectral radius} of $R$, where $\{ \lambda_1, \dots ,\lambda_{d^2} \}$  are the (complex) eigenvalues of $R$ arranged in decreasing order of magnitude. Then 
\begin{itemize}
\item[(i)]
$r(R)$ is an eigenvalue of $R$, and it has a positive eigenvector.

\vspace{1mm}

\item[(ii)]
If additionally, $R$ is unit preserving then $r(R)=1$ with eigenvector $\mathbf{1}$.

\vspace{1mm}

\item[(iii)]
If additionally, $R$ is irreducible then $r(R)$ is a nondegenerate eigenvalue for $R$ and $R_*$, and both corresponding eigenvectors are strictly positive.

\vspace{1mm}
 
\item[(iv)] If additionally, $R$ is primitive  then $|\lambda_i| <r(R)$ for all eigenvalues other than $r(R)$.
\end{itemize}

\end{thm}
As a corollary, if the Markov transition operator $T$ is irreducible then it has a unique full rank stationary state (i.e. $T_{*}(\rho_{\text{ss}}) =\rho_{\text{ss}}$), and if $T$ is also primitive then any state 
converges to the stationary state in the long run (mixing or ergodicity property);
$$
\lim_{n\to\infty} (T_*)^n (\rho) = \rho_{\text{ss}},  %\quad \text{ for all }\rho\in \mathcal{S}(\mathcal{H}).
$$
or in the Heisenberg picture
\begin{equation*}\label{statlimit}
\lim_{n\rightarrow \infty} T^n(A)= {\rm tr}[A\rho_{\text{ss}}]\mathbf{1}.%\quad \text{ for all }X\in \lh.
\end{equation*}
For reader's convenience we state here a lemma \cite{Hiai2007} which will be used in the proof of the main theorem.

\begin{lem}[\cite{Hiai2007}]\label{thm:sanovlem} Let $\mathcal{A}$ be a finite-dimensional $C^{\ast}$-algebra, let $R$ be a positive linear map on $\mathcal{A}$ and suppose that $R$ has a strictly positive eigenvector. Then for any state $\varphi : \mathcal{A} \to \mathbb{C}$ and any positive $X \in \mathcal{A}$ 
\begin{equation*}
	\lim_{n \rightarrow \infty} \tfrac{1}{n} \log \varphi \left( R^{n}(X)  \right) = \log r(R).
\end{equation*}
\end{lem}

We now state and prove our main result.

\begin{thm}[\textbf{Sanov theorem for quantum Markov chains}]\label{thm:sanov}
Consider a quantum Markov chain on $M_{d}$ whose transition operator $T$ is primitive. Then the the level $m$ empirical measure \eqref{eq.empirical.measure} satisfies a large deviations principle on $\mathbb{R}^{k^{m}}$. The rate function is the Legendre transform of $\log r(\tilde{T}_{t,m})$, where $r(\tilde{T}_{t,m})$ is the spectral radius of a certain restriction 
$\tilde{T}_{t,m}$ of the extended transition operator defined in \eqref{eq:sanovmaps}.
\end{thm}

\begin{proof}

By the Gartner-Ellis Theorem \ref{th.ge} it suffices to show that the limit
\begin{equation}\label{eq:proofge}
F^{(m)} (t) =\lim_{n\to\infty} \tfrac{1}{n} \log \Gamma_{n}^{(m)}(t)
\end{equation}
exists for all $t$ and $F^{(m)}$ is a continuous function. As it is defined, $T_{t,m}$ may not satisfy the conditions of Lem. \ref{thm:sanovlem}, but we will show that \eqref{eq:proofge} holds the same when $T_{t,m}$ is replaced by a certain restriction $\tilde{T}_{t,m}$ which does satisfy the conditions.

\textbf{1. Invariance.} Let $\mathcal{B}_{m}$ be the (non-unital) subalgebra of $\mathcal{D}^{m,k}_{d}$ given by
\begin{equation*}
    \mathcal{B}_{m} = \bigoplus_{i_{1},\ldots,i_{m-1}}^{k} Q_{i_{1},\ldots,i_{m-1}} M_{d} Q_{i_{1},\ldots, i_{m-1}} 
\end{equation*}
where $Q$ is the projection onto the support of $M^{(m)}$. We will show that 
$\mathcal{B}_{m}$ is invariant under $T_{t,m}$, i.e.
\begin{equation*}
    \left[ T_{t,m}(Y)\right]_{i_{1},\ldots,i_{m-1}} \in Q_{i_{1},\ldots i_{m-1}}M_{d} Q_{i_{1},\ldots,i_{m-1}}
\end{equation*}
for every $Y \in \mathcal{B}_{m}$. For this it suffices to show that
\begin{equation*}
    V_{i_{1}}^{\ast} \left[ Y \right]_{i_{2}, \ldots, i_{m}} V_{i_{1}} (u) = 0
\end{equation*}
for all $u \in \text{ker} \left[ M^{(m)} \right]_{i_{1},\ldots,i_{m-1}},$ or equivalently (since $\text{ker}(A^{\ast}A) = \text{ker}(A)$) all $u \in \mbox{ker}(V_{i_{m-1}} \cdots V_{i_{1}})$. Now
\begin{equation*}
    \left[Y\right]_{i_{2},\ldots,i_{m}}(v) = 0 \quad \text{for all }v \in \text{ker}(V_{i_{m}}\cdots V_{i_{2}})
\end{equation*}
and we either have $u \in \text{ker}(V_{i_{1}})$ or $V_{i_{1}}(u) \in \text{ker}(V_{i_{m-1}}\cdots V_{i_{2}})$ we conclude that for every $u \in \text{ker}\left[Y \right]_{i_{1}, \ldots, i_{m-1}}$ we have $u \in \text{ker}(\left[T(Y)\right]_{i_{1}, \ldots, i_{m-1}})$, proving that $T_{t,m}$ leaves $\mathcal{B}_{m}$ invariant.

Let us denote by $\tilde{T}_{t,m}$ the restriction of $T_{t,m}$ to $\mathcal{B}_{m}$. Then, since 
$M^{(m)}\in\mathcal{B}_{m}$ the moment generating function can be expressed as 
$$
\Gamma_{n}^{(m)}(t) 
= \langle \psi \vert \sum_{i_{1}, \ldots, i_{m-1} =1}^{k} \left[ \tilde{T}_{t,m}^{n-m+1}\left(M^{(m)}\right)\right]_{i_{1}, \ldots, i_{m-1}} \vert \psi \rangle
$$

\textbf{2. Primitivity.} Since $\tilde{T}_{t,m} > c \tilde{T}_{0,m}$ for some positive constant $c$, it suffices to show that 
%$\tilde{T}_{0,m}$ is primitive, 
%We prove that $\tilde{T}_{0,m}$ is irreducible by showing that 
there exists $n \in \mathbb{N}$ such that, for any $X \in \mathcal{B}_{m}$,
\begin{equation*}
    \tilde{T}_{0,m}^{n}(X) \geq c^\prime \mathbf{1}_{\mathcal{B}}
\end{equation*}
for some $c^\prime > 0$. We can assume that $n > m$, in which case 
\begin{align*}
    &\left[ \tilde{T}_{0,m}^{n}(X)\right]_{i_{1},\ldots,i_{m-1}} \\ &=\sum_{i_{m}, \ldots, i_{m+n-1} = 1}^{k} V_{i_{1}}^{\ast} \cdots V_{i_{n}}^{\ast} \left[ X \right]_{i_{n+1}, \ldots, i_{m+n-1}} V_{i_{n}} \cdots V_{i_{1}}\\
    &= V_{i_{1}}^{\ast} \cdots V_{i_{m-1}}^{\ast} \left( \sum_{i_{m},\ldots, i_{n} = 1}^{k} V_{i_{m}}^{\ast} \cdots V_{i_{n}}^{\ast} \hat{X} V_{i_{n}} \cdots V_{i_{m}} \right) V_{i_{m-1}} \cdots V_{i_{1}}
\end{align*}
where $\hat{X}$ is the sum of the blocks of $X$ given by
\begin{equation*}
    \hat{X} = \sum_{i_{1}, \ldots, i_{m-1}=1}^{k} \left[ X \right]_{i_{1}, \ldots, i_{m-1}}.
\end{equation*}
The remaining sum can be written in terms of the original transition operator $T$ associated to the quantum Markov chain; recall that for $Y \in M_{d}$
\begin{equation*}
    T (Y) = \sum_{i=1}^{k} V_{i}^{\ast} Y V_{i}
\end{equation*}
from which we obtain the expression
\begin{equation*}
    \left[ \tilde{T}_{0,m}^{n}(X)\right]_{i_{1},\ldots,i_{m}} = V_{i_{1}}^{\ast} \cdots V_{i_{m-1}}^{\ast} T^{n-m+1}(\hat{X}) V_{i_{m-1}} \cdots V_{i_{1}}.
\end{equation*}
Since the original Markov chain is assumed to be primitive, there exists $r \in \mathbb{N}$ such that $T^{r}(\hat{X}) \geq c \mathbf{1}$ for some $c > 0$. Therefore, with $n \geq r+m-1$,
\begin{align*}
    \left[ \tilde{T}_{0,m}^{n}(X)\right]_{i_{1}, \ldots, i_{m-1}} &\geq c V_{i_{1}}^{\ast} \cdots V_{i_{m-1}}^{\ast} V_{i_{m-1}} \cdots V_{i_{1}}\\
        &\geq c^\prime Q_{i_{1}, \ldots, i_{m-1}}\\
        &= c^\prime \left[ \mathbf{1}_{\mathcal{B}} \right]_{i_{1},\ldots,i_{m-1}}.
\end{align*}
%Since $M^{(m)} \in \mathcal{B}_{m}$ we can compute the moment generating functions on this restriction i.e. $\Gamma_{n}(t) = \varphi \left( \tilde{T}_{t,m}^{n}(V^{(m)})\right)$. 

Using the invariance and primitivity property, we can apply  Lemma \ref{thm:sanovlem}, to find that the limiting moment generating function is given by  
\begin{equation*}
F^{(m)} (t) =\lim_{n\to\infty} \tfrac{1}{n} \log \Gamma_{n}^{(m)}(t) = r( \tilde{T}_{t,m}).
\end{equation*}
Moreover, since $\tilde{T}_{t,m}$ is an analytic perturbation in $t$ of $\tilde{T}_{0,m}$, the spectral radius $t \mapsto r(T_{t,m})$ is a smooth function \cite{Kato1976}, so the large deviation principle follows from the G\"{a}rtner-Ellis theorem. 

%
%To prove that $T_{t,m}$ satisfies the conditions of Lem. \ref{thm:sanovlem} we employ a non-commutative Perron-Frobenius theorem for irreducible and aperiodic positive linear maps \cite[Thms. 4.3-4.4]{Evans1978}: for any irreducible and aperiodic positive map $T$ on a finite dimensional $C^*$ algebra, the spectral radius $r(T)$ is a simple eigenvalue whose corresponding eigenvector is strictly positive. In the second part of the proof we argue that the map $T_{t,m}$ is an irreducible and aperiodic positive map; since $T_{t,m}$ is an analytic perturbation in $t$ of $T_{0,m}$, the spectral radius $t \mapsto r(T_{t,m})$ is a continuous function \cite{Kato1976}, and through Lemma \ref{thm:sanovlem} we obtain the desired large deviation principle from the G\"{a}rtner-Ellis theorem. 
%
%
%The proof consists of showing that $T_{t,m}$ is positive and irreducible.  Since a block diagonal matrix is (strictly) positive if and only if each block is (strictly) positive, it is straightforward to verify that $T_{t,m}$ is positive. For irreducibility, 

\end{proof}

%\subsection{Note on transition operator}
%
%We briefly clarify why at $t=0$ the dominant eigenvalue of our map $T_{t,m}$ is $1$. The idea is that our $T_{t,m}$ is related by a similarity transformation to another transition operator which more clearly has the desired properties at $t=0$. We define a map $T^{0}_{m}$ on $\mathcal{B}_{m}$ which acts on each block separately as
%\[
%    \left[ T^{0}_{m}(Y) \right]_{i_{1},\ldots,i_{m-1}} = V_{i_{1}}^{\ast} \cdots V_{i_{m-1}}^{\ast} Y_{i_{1}, \cdots, i_{m-1}} V_{i_{m-1}} \cdots V_{i_{1}}.
%\]
%Assuming for simplicity that the Kraus operators are invertible,  we define another map $\bar{T}^{0}_{m}$ by
%\[
%    \left[ \bar{T}^{0}_{m}(Y) \right]_{i_{1},\ldots,i_{m-1}} = (V_{i_{m-1}}^{*})^{-1} \cdots (V_{i_{1}}^{*})^{-1} Y_{i_{1}, \cdots, i_{m-1}} V_{i_{1}}^{-1} \cdots V_{i_{m-1}}^{-1 }
%\]
%then $\bar{T}^{0}_{m}$ is the inverse of $T^{0}_{m}$. Finally, we define the map $\hat{T}_{t,m}$ by writing
%\[
%    \left[ \hat{T}_{t,m}(Y) \right]_{i_{1},\ldots,i_{m-1}} = \left[ \bar{T}^{0}_{m}  \circ T_{t,m} \circ T^{0}_{m} (Y) \right]_{i_{1},\ldots,i_{m-1}};
%\]   
%then
%\[
%    \left[ \tilde{T}_{t,m}(Y) \right]_{i_{1},\ldots,i_{m-1}} = \left[ T_{t,m}(Y) \right]_{i_{1},\ldots,i_{m-1}}.
%\]
%We may conclude that $T_{t,m}$ and $\tilde{T}_{t,m}$ are identical up to a similarity transformation; the latter map is identity preserving for $t=0$, which means the dominant eigenvalue at $t=0$ is indeed $1$.
%
%\subsection{Central Limit Theorem}\label{sec.clt}

As a corollary to our main result, we establish a central limit theorem for each of the empirical measures.

\begin{col}\label{cor:clt}
Let $\set{\hat{\mathbb{P}}^{(m)}_{n}}$ be the sequence of distributions of the empirical measure of length $m$ defined in Eq. \eqref{eq.empirical.measure}. Then $\set{\hat{\mathbb{P}}^{(m)}_{n}}$ satisfies the Central Limit Theorem: that is (cf. Eq. \eqref{eq.ld.sanov}), as $n \rightarrow \infty$,
\begin{equation*}
\sqrt{n} \left( \hat{\mathbb{P}}^{(m)}_{n} -  p^{(m)}\right) \overset{\mathcal{D}}{\longrightarrow} N(0,V^{(m)})
\end{equation*}
where $\mathcal{D}$ denotes convergence in distribution. Here $p^{(m)}$ and $V^{(m)}$ are the mean and variance with respective components
\begin{equation*}
p^{(m)}_{i} = \left. \frac{\partial \log r(\tilde{T}_{t,m}) }{\partial t_{i}}\right|_{t=0}, \qquad V^{(m)}_{i,j} = \left. \frac{\partial^{2} \log r(\tilde{T}_{t,m}) }{\partial t_{i} \partial t_{j}}\right|_{t=0}, \quad 1 \leq i,j \leq k^m.
\end{equation*}
\end{col}
\begin{proof}
Our main result relied on the convergence of the logarithmic moment generating functions $ \tfrac{1}{n} \log \Gamma_{n}^{(m)}(t)$ in Eq. \eqref{eq:proofge} to $F^{(m)} (t)$. For the purpose of establishing an LDP, it was sufficient to consider only real values for the parameter $t$. However, the same analytic perturbation arguments can be used to prove that locally in complex neighbourhood of $t=0$, the Eq. \eqref{eq:proofge} holds and the limiting function $F^{(m)} (z)$ is analytic in that region. The Central Limit Theorem for $\set{\hat{\mathbb{P}}^{(m)}_{n}}$ is then a consequence of the main result of \cite{Letters1993}.

\end{proof}
We note that since the sequence $\set{\hat{\mathbb{P}}^{(m)}_{n}}$ is effectively a sequence of probability distributions on the probability simplex of dimension $k^m -1$, the variance $V^{(m)}$ is degenerate in that at least one of its eigenvalues vanishes, corresponding to the degree of freedom in $\mathbb{R}^{k^{m}}$ orthogonal to the probability simplex.

We end this section with a brief note on what happens when the transition operator is not irreducible. The G\"{a}rtner-Ellis theorem relies on differentiability of the logarithmic moment generating function $F(t)$ defined in Eq. \ref{eq:gelimit} (an hypothesis which may be weakened to smoothness of $F(t)$ in a neighbourhood of the origin). We use irreducibility to ensure that the spectral radius, and therefore the logarithmic moment generating function, is a differentiable function. 

If $F(t)$ is not differentiable at $t=0$, the first moments obtained as $\partial_{t} F(t)$ at $t=0$ are not well defined, corresponding to a breaking down of a law of large numbers. On the level of the transition operator this corresponds to the dynamics breaking up into two disjoint parts, each with its own logarithmic moment generating function. In this case an LDP may still hold, but with a non-convex rate function. The rate function obtained from the G\"{a}rter-Ellis theorem is a Legendre transform, and is the convex envelope of the actual rate function. This non-convexity of the large deviations rate function is associated to dynamical phase transitions in statistical mechanics \cite{Ellis1995,Touchette2009} and more recently in open quantum systems \cite{Garrahan2011,Lesanovsky2013,VanHorssen2014}.

\section{Examples}\label{sec.examples}
Here we describe two examples illustrating the mathematical results.
The first example is a quantum Markov chain where the large deviations rate function associated to the empirical measure and pair empirical measure shows dependence on some physical parameter. In the second example, the empirical measure rate function shows no dependence on a physical parameter, but the pair empirical measure rate function does. This example shows that, in order to uncover dynamical phase transitions through non-analyticities in a large deviations rate function, it may be necessary to consider higher-level statistics.

\subsection*{Example 1}

\begin{figure}
    \subfloat[]
    {\includegraphics[width=.6\textwidth]{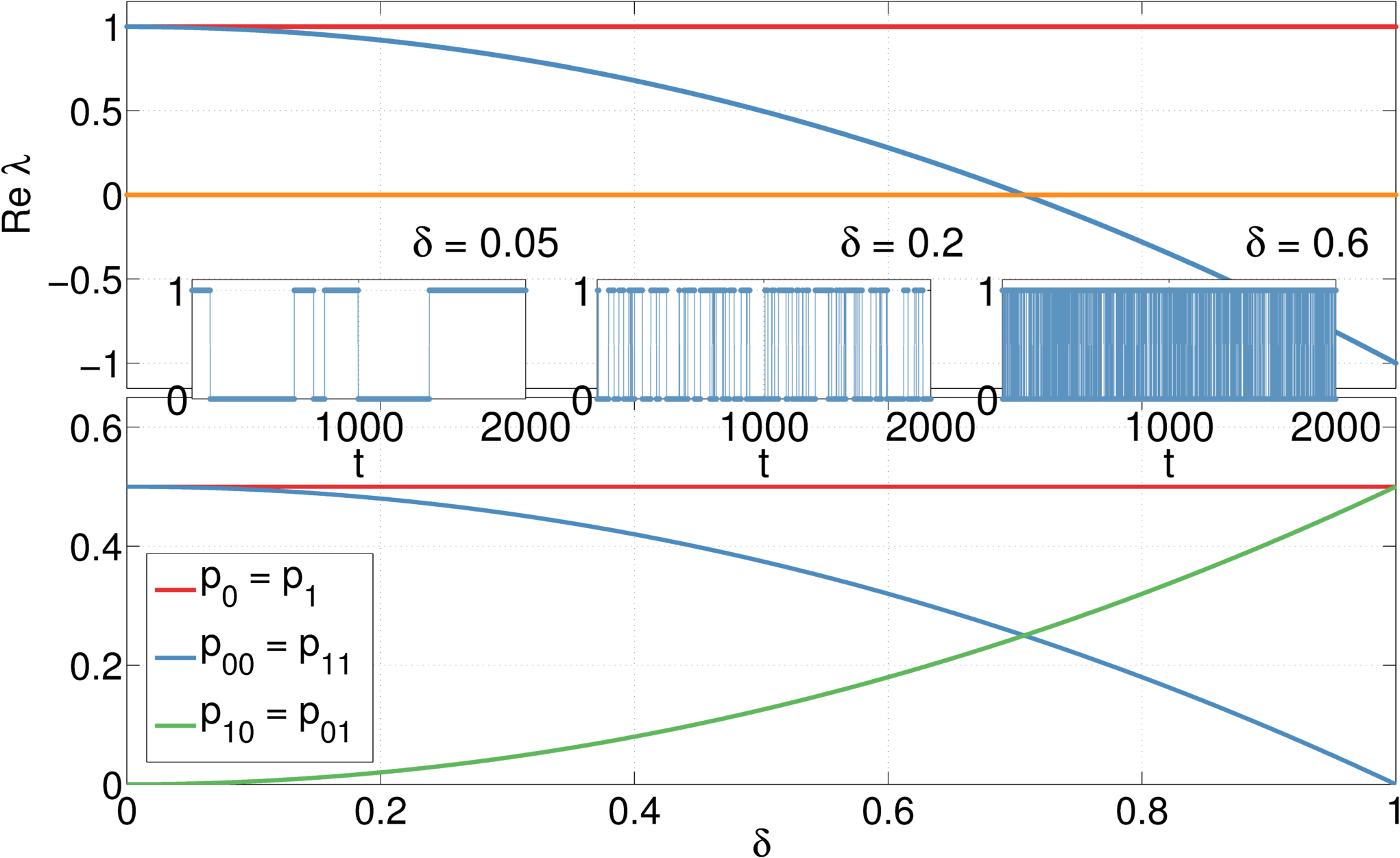}} \
%    \subfloat[$\partial_{t} r(T_{t,1})$]   
%    {\includegraphics[width=.25\textwidth]{sanov-ex1-rate2}}
%    \subfloat[$\partial_{t} r(T_{t,2})$]   
%    {\includegraphics[width=.25\textwidth]{sanov-ex1-rate4}}
    \subfloat[]   
    {\includegraphics[width=.38\textwidth]{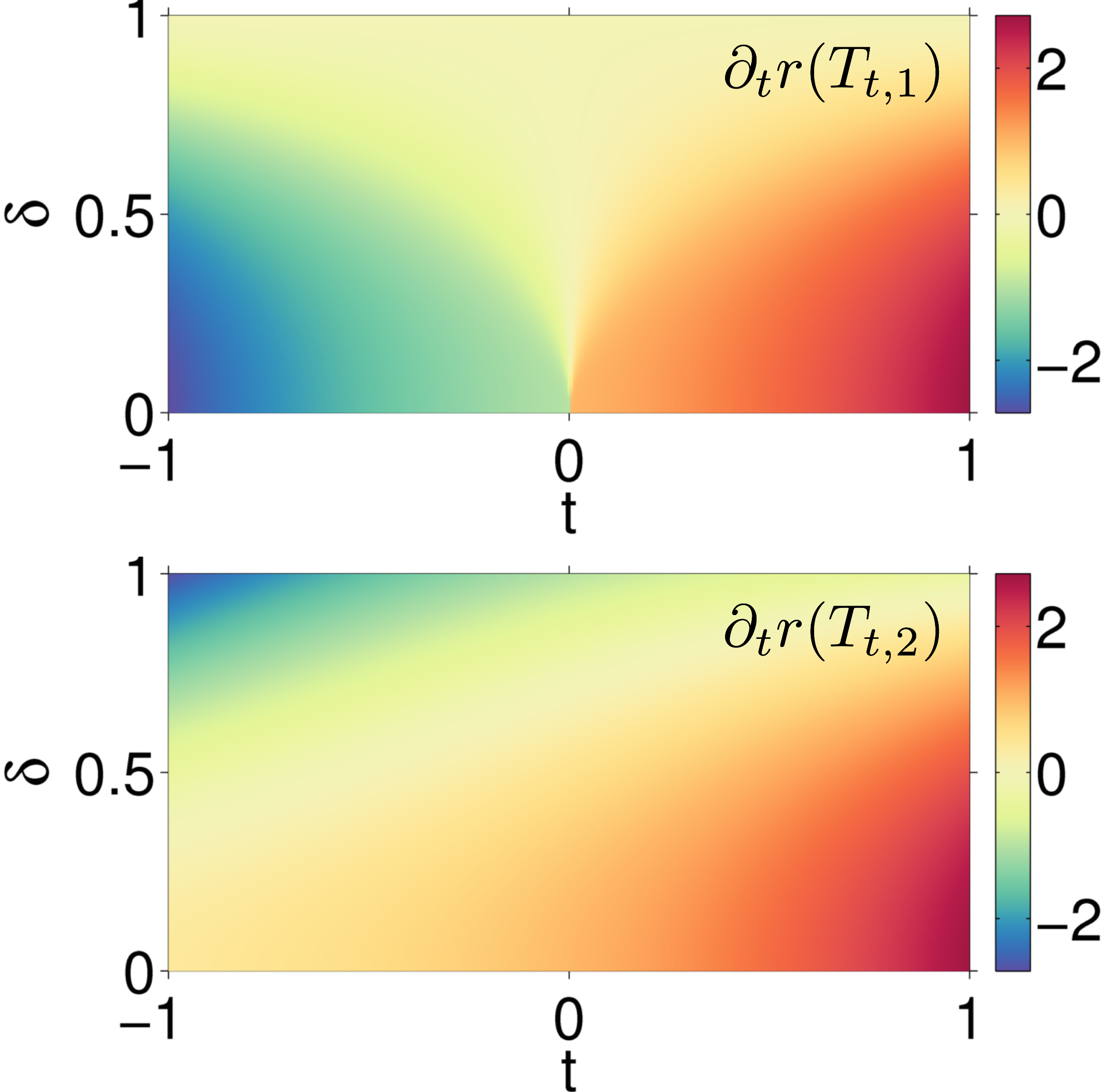}}     
    \caption{\textbf{(a):} Eigenvalues $\lambda$ of $T$ (top), statistics of trajectories (bottom):  while the probabilities $p_{k}$ are equal and constant, the pair probabilities $p_{i,j}$ vary with $\delta$, which is reflected in the jump trajectories (see inset); \textbf{(b):} Derivatives of Sanov theorem level 1 and 2 spectral radii as functions of $\delta$ and the LD parameter with parametrisations $t [-1,1] \in \mathbb{R}^{2}$ and $t [1, -1; -1, 1] \in M_{2}$, respectively. Note the discontinuity in $\partial_{t} r(T_{t,1})$ at $\delta=0$ where the LDP no longer holds.  }
    \label{fig:ex1}
\end{figure}

Consider the two-dimensional quantum Markov chain with transition operator $T$ acting on a density matrix $\rho \in M_{2}$ as
\begin{equation}\label{eq:ex1}
    T(\rho) = V_{0} \rho V_{0}^{\ast} + V_{1} \rho V_{1}^{\ast}
\end{equation}
where the Kraus operators are given by
\begin{equation*}\label{eq:example1}
    V_{0} = \mat{0}{\delta}{0}{\epsilon},\quad V_{1} = \mat{\epsilon}{0}{\delta}{0}
\end{equation*}
where $\epsilon = \sqrt{1-\delta^{2}}$ with $0 \leq \delta \leq 1$. Defining $\ket{u} = \delta \ket{0} + \epsilon \ket{1}$ and $\ket{d} = \epsilon \ket{0} + \delta \ket{1}$, the Kraus operators can be expressed as $V_{0} = \ket{u} \bra{1}$, $V_{1} = \ket{d} \bra{0}$; the parameter $\delta$ interpolates between a trivial process at $\delta = 0$, where the Kraus operators project onto the eigenstates of $\sigma_{z}$, and the cyclic process with  $V_{0} = \ket{0}\bra{1} = V_{1}^{\ast}$ at $\delta = 1$. The eigenvalues of $T$ are plotted in Fig. \ref{fig:ex1}, showing that the set of eigenvalues reduces to $\set{0,1}$ at $\delta=0$ and $\set{-1,0,1}$ at $\delta=1$.

Considering the level 1 and 2 statistics in the context of our Sanov theorem, in Fig. \ref{fig:ex1} we have plotted the probabilities $p_{k}$ and $p_{i,j}$ to obtain $X^{(n)} = k$ and $(X^{(n)},X^{(n+1)}) = (i,j)$, respectively, along the output trajectory in the stationary regime. By Thm. \ref{thm:sanov}, the empirical measures associated to these jump statistics satisfy an LDP for  $\delta > 0$, with rate function computed in Eq. \eqref{eq:gerate} as the Legendre-Fenchel transformation of the spectral radius of $r(T_{t,m})$ of the associated transition operator. The first moment of the level $k$ empirical measure is computed as the derivative $\partial_{t} r(T_{t,k})$ evaluated at $t=0$; Fig. \ref{fig:ex1} shows the derivatives of the spectral radii $\partial_{t} r(T_{t,1})$ and $\partial_{t} r(T_{t,2})$. In this case, both the level 1 and 2 spectral radii (and therefore the rate functions) show a dependence on $\delta$, even though the first moment $\partial_{t} r(T_{t,1}) \vert_{t=0}$ is constant.

\subsection*{Example 2}

We now consider a quantum Markov chain which also satisfies Thm. \ref{thm:sanov}, but where the level 1 statistics are independent of the physical parameter of the model. Let $\rho \in M_{2}$ be a density matrix and define the transition operator $T$ as in Eq. \eqref{eq:ex1}, but where the Kraus operators are now defined as
\begin{equation*}
    V_{0} = \frac{1}{\sqrt{2}}\mat{1}{0}{i \sin \omega}{\cos \omega},\quad V_{1} = \frac{1}{\sqrt{2}} \mat{\cos \omega}{i \sin \omega}{0}{1}
\end{equation*}
where $0 \leq \omega \leq 2\pi$ (these dynamics may be obtained from a particular choice of parameters in a Heisenberg XYZ interaction between each noise atom and the system.)

Considering how the stationary states change with $\omega$ (see Fig. \ref{fig:ex2}) we note that for $\omega=0$, $T$ becomes the identity map with full degeneracy of the eigenvalue $1$. For $\omega \ll 1$ perturbation of the degenerate eigenvalue $1$ shows that the eigenvalues split into $\lambda_{1} = 1$, $\lambda_{2} \approx 1 - \omega^{2}/2$ and a pair of complex conjugate eigenvalues $\lambda_{3} \approx 1 + i \omega$, $\lambda_{4} \approx \bar{\lambda}_{3}$. For $0 < \omega < \pi$ the stationary state $\rho_{\text{ss}}$ is unique and a multiple of the identity, $\rho_{\text{ss}} = \tfrac{1}{2}\textbf{1}.$ At the point $\omega = \pi$ the Kraus operators are unitarily equivalent and stationary states are of the form $p \ket{0}\bra{0} + (1-p) \ket{1} \bra{1}$ with $0 \leq p \leq 1$.

\begin{figure}
    \subfloat[]
    {\includegraphics[width=.6\textwidth]{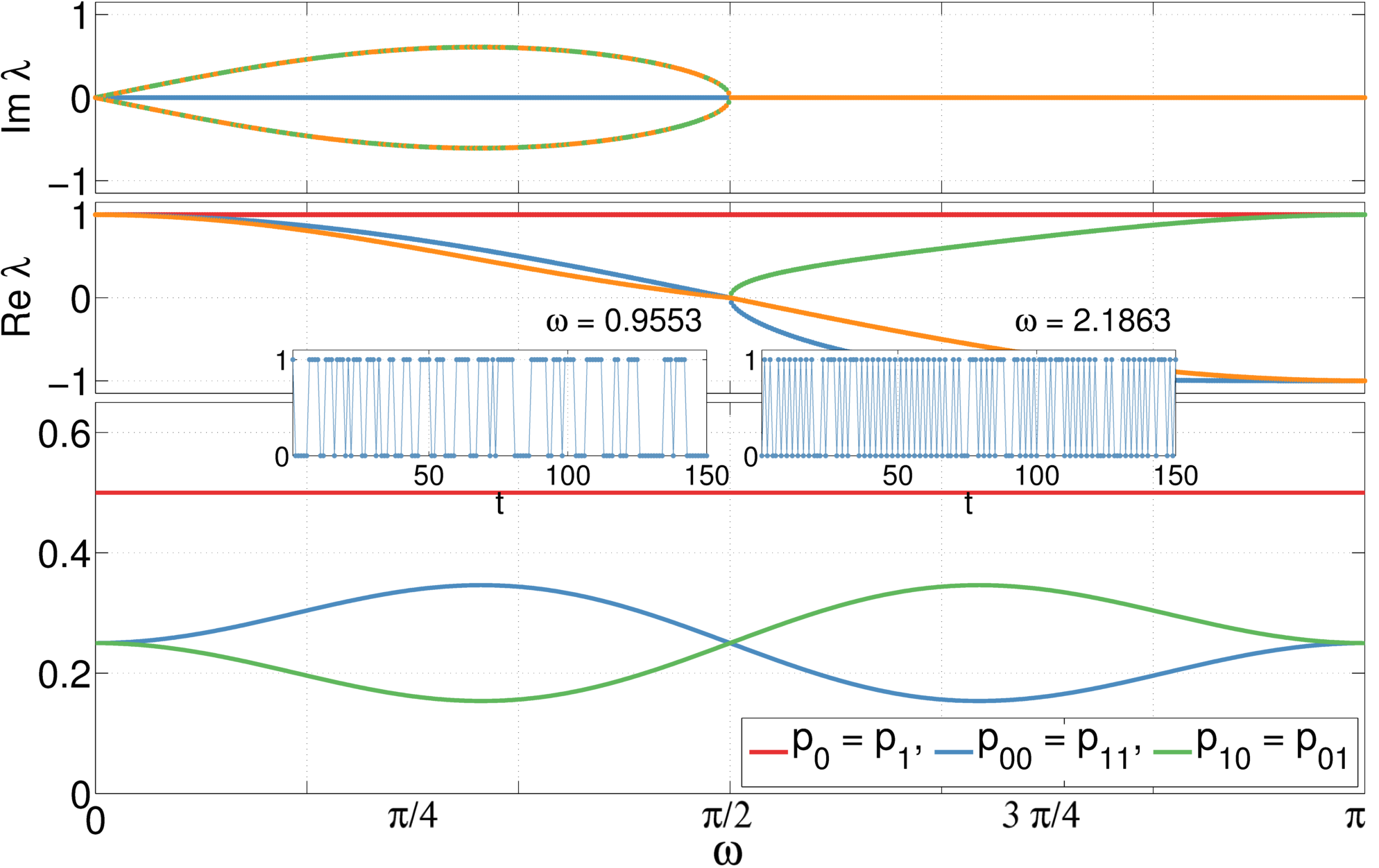}}\
%    \subfloat[$\partial_{t} r(T_{t,1})$]   
%    {\includegraphics[width=.25\textwidth]{sanov-ex2-rate1}}
%    \subfloat[$\partial_{t} r(T_{t,2})$]
%    {\includegraphics[width=.25\textwidth]{sanov-ex2-rate4}}     
    \subfloat[]   
    {\includegraphics[width=.38\textwidth]{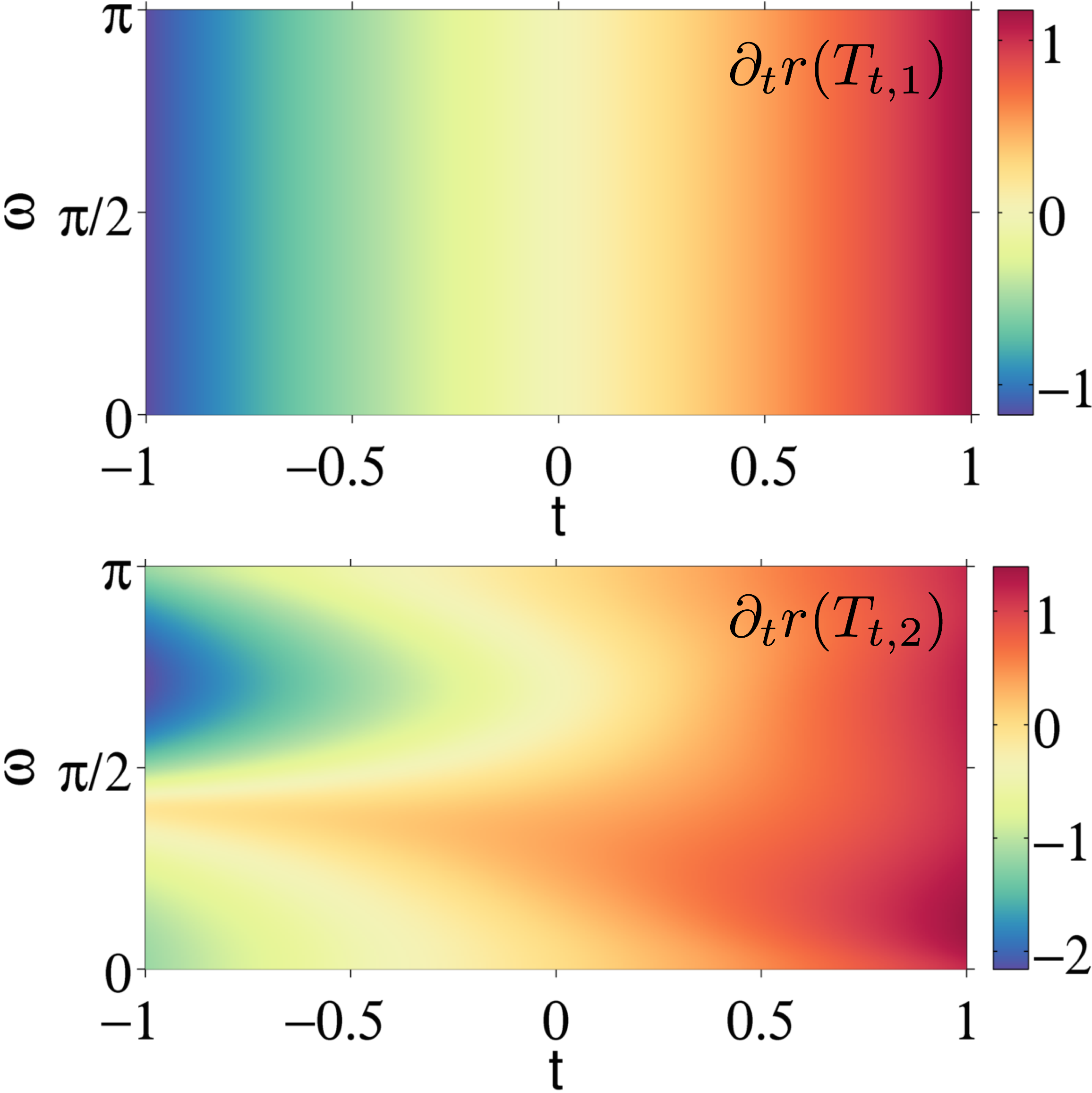}}      
    \caption{\textbf{(a):} Eigenvalues $\lambda$ of $T$ (top), statistics of trajectories (bottom) and jump trajectories (inset); \textbf{(b):} Derivatives of Sanov theorem level 1 and 2 spectral radii as functions of $\delta$ and the LD parameter with same parametrisation as in Fig. \ref{fig:ex1}. }
    \label{fig:ex2}
\end{figure}

As in the previous example, we consider the jump probabilities $p_{k}$ and $p_{i,j}$ in the stationary regime (see Fig. \ref{fig:ex2}). For $0 < \omega < \pi$ the probabilities $p_{k}, k=0,1$ are independent of $\omega$ with
\begin{equation*}
    p_{k} = \Tr \left( V_{k} \rho_{\text{ss}} V_{k}^{\ast} \right) = \tfrac{1}{2} \Tr \left( V_{k}^{\ast} V_{k} \right) = \tfrac{1}{2}
\end{equation*}
and similarly we obtain 
\begin{eqnarray*}
    p_{0,0} &= \tfrac{1}{4}\left(1 - \sin^{2} \omega \cos \omega \right) = p_{1,1},\\
    p_{0,1} &= \tfrac{1}{4}\left(1 + \sin^{2} \omega \cos \omega \right) = p_{1,0};
\end{eqnarray*}
as shown in Fig. \ref{fig:ex2} this dependence on $\omega$ is reflected in the output trajectories, with increased intermittency when $p_{0,1} > p_{0,0}$.

By Thm. \ref{thm:sanov}, the empirical measure associated to the level 1 and level 2 statistics on the output of this quantum Markov chain satisfies an LDP, with rate functions obtained from the corresponding spectral radii $r(T_{t,k})$. As shown in Fig. \ref{fig:ex2}, the level 2 spectral radius $r(T_{t,2})$ depends on $\omega$, while $r(T_{t,1})$ is constant, as we will now show.

\begin{lem}For the quantum Markov chain defined by $T$, the large deviations rate function associated to the level 1 empirical measure (i.e. sample mean) of the output process is independent of $\omega$, for $0<\omega<\pi$.
\end{lem}
\begin{proof}
At the point $\omega = \pi/2$ the trajectories of the output process are equivalent to those of a classical process of i.i.d. fair coin tosses; this becomes evident by expressing the Kraus operators as
\begin{eqnarray*}
    V_{0} &= \ket{u}\bra{0},\quad \ket{u} = \tfrac{1}{\sqrt{2}} \left(\ket{0} + i \ket{1} \right),\\
    V_{1} &= \ket{d}\bra{1},\quad \ket{d} = \tfrac{1}{\sqrt{2}} \left(i \ket{0} + \ket{1} \right)
\end{eqnarray*}
which always project onto the states $\ket{u}$ and $\ket{d}$ with equal probability. The large deviations rate function $I$ associated to the sample mean of i.i.d. fair coin tosses is \cite{DenHollander2000}
\begin{equation*}
    I(x) = \log 2 - x \log x - (1-x) \log (1-x),\quad 0 < x < 1.
\end{equation*}
By Thm. \ref{thm:sanov} the spectral $r(T_{t,1})$ is related to the rate function $I$ by a Legendre transformation,
\begin{align*}
    \log r(T_{t,1}) &= \sup_{0 < x < 1} \set{ x t - I(x)}\\
        &= t - \log 2\left(1 + e^{-t}\right)
\end{align*}
and so $r(T_{t,1}) = \frac{1}{2} \left( e^{t} + 1 \right)$. 

For $\omega \neq \pi/2$ it is easy to check that $\lambda_{t} = \frac{1}{2} \left( e^{t} + 1 \right)$ remains an eigenvalue of $T_{t,1}$ with eigenmatrix
\begin{equation*}
    \rho_{\text{ss}} + \frac{1}{2}\left( \lambda_{t} - 1 \right)\mat{1+\cos \omega}{i \sin \omega}{-i \sin \omega}{1 - \cos \omega}
\end{equation*}
where $\rho_{\text{ss}}$ is the stationary state $\tfrac{1}{2}\textbf{1}$. Since $\lambda_{t}$ is independent of $\omega$ and $\lambda_{t} \rightarrow 1$ as $t \rightarrow 0$ we conclude that the moments $\partial_{t}^{n} \lambda_{t} \vert_{t=0}$ are independent of $\omega$.
\end{proof}

This example shows that the LD rate functions obtained from Thm. \ref{thm:sanov} are useful in uncovering dynamical behaviour of a system which is not immediately obtained from the lowest level LD picture.

\section{Discussion}

We have shown that a large deviations principle holds for  the empirical measure associated to an arbitrary number of subsequent outcomes obtained by measuring the output of a primitive quantum Markov chain. This extends the $m=1$ large deviation result for the total counts of outcomes obtained in \cite{Hiai2007}, which is the basis of the thermodynamics theory of quantum trajectories, and the theory of dynamical phase transitions \cite{Garrahan2011,Lesanovsky2013}. We presented an example in which the $m=1$ LD rates are constant with respect to a system parameter, while the $m=2$ theory captures this dependence. This suggests that a continuous-time version of our result would be relevant for a better understanding of dynamical phase transitions. Another direction in which the work can be extended is towards a Donsker-Varadhan LD theory for the empirical process of infinite trajectories.  

Additionally, we showed that the empirical measure satisfies the Central Limit Theorem, extending the result from \cite{Guta2011} which dealt with total counts statistics. The result, and its extensions to more general collective variables of the output are directly relevant for the statistical theory of system identification of open systems \cite{Guta2014}.

%*****************************************

\begin{acknowledgments}
\emph{Acknowledgments}.--- The authors would like to thank Juan Garrahan and Igor Lesanovsky for fruitful discussions. 
This work was supported by the EPSRC grant EP/J009776/1. 
\end{acknowledgments}

%*****************************************

%
%%% BIBLIOGRAPHY %%%
%
%\bibliographystyle{utphys}
%\bibliographystyle{apsrev4-1}
\bibliography{bibliography}

%\appendix
%\section*{Appendix}

\end{document}